\documentclass[runningheads]{llncs}
\usepackage{amssymb}
\usepackage{amsmath,amssymb}
\usepackage{lscape, rotating}
\usepackage{graphicx}
\usepackage{extarrows}
\usepackage{array}

\usepackage{arydshln}

\newtheorem{rem}{\textbf{Remark}}

\newtheorem{no}{\textbf{Note}}

\newtheorem{exa}{\textbf{Example}}

\begin{document}

\pagestyle{plain}

\mainmatter

\title{Improved Upper Bounds on Systematic-Length for Linear Minimum Storage Regenerating Codes
\protect\footnote{$^1$Kun Huang and Ming Xian are with State Key Laboratory of Complex Electromagnetic Environment Effects on Electronics and Information System, National University of Defense Technology, Changsha, 410073, China \email{(khuangresearch923@gmail.com; qwertmingx@tom.com)}.\\
$^2$ Udaya Parampalli is with Department of Computing and Information Systems, University of Melbourne,  VIC 3010, Australia \email{(udaya@unimelb.edu.au)}.\\
}
}
\titlerunning{}

\author{
Kun Huang\inst{1},\ Udaya Parampalli\inst{2}  \and\ Ming Xian\inst{1}
}

\authorrunning{}

\institute{}

\maketitle

\begin{abstract}
In this paper, we revisit the problem of finding the longest systematic-length $k$ for a linear minimum storage regenerating (MSR) code with optimal repair of only systematic part, for a given per-node storage capacity $l$ and an arbitrary number of parity nodes $r$. We study the problem by following a geometric analysis of linear subspaces and operators. First, a simple quadratic bound is given, which implies that $k=r+2$ is the largest number of systematic nodes in the \emph{scalar} scenario. Second, an $r$-based-log bound is derived, which is superior to the upper bound on log-base $2$ in the prior work. Finally, an explicit upper bound depending on the value of $\frac{r^2}{l}$ is introduced, which further extends the corresponding result in the literature.

\end{abstract}

\section{INTRODUCTION}

Distributed storage systems (DSSs) are efficient storage systems designed for coping with the tremendously growing data. They have become an indispensable component of many distributed applications such as cloud storage, social networking and peer to peer networking. Through introducing redundancy in the form of replication or erasure coding, DSSs provide data storage services with high reliability \cite{Re:A.G.Dimakis1}. Compared with the systems that employ replication, erasure codes achieve higher reliability for the same level of redundancy \cite{Re:H.Weatherspoon}. With regard to repair bandwidth as a classical measurement of reliable DSSs, Dimakis et al. \cite{Re:A.Dimakis} propose a novel class of codes termed \emph{regenerating codes} that are particularly advantageous in the repair efficiency.

\subsection{Regenerating Codes}

Regenerating codes \cite{Re:A.Dimakis} are a family of maximal distance separable (MDS) codes. In this framework, a $B$-sized original data file over a finite field $\mathbb{F}$ is encoded into $n\alpha$ symbols that are distributed across $n$ nodes with each node storing $\alpha$ symbols. Data symbols stored in any $k$ out of $n$ nodes are sufficient to recover the original data file, and meanwhile, any single node failure can be repaired by downloading $\beta$ symbols from each of $d$ helper nodes out of the remaining $n-1$ nodes. It is shown in \cite{Re:A.Dimakis} that regenerating codes with a parameter set $\{n,k,d,\alpha,\beta,B\}$ have to comply with the following constraint (the tradeoff curve):
\begin{equation}\label{cut}
B\leq \sum_{i=1}^{k}\min\{\alpha,(d-i+1)\beta\},
\end{equation}
where codes that achieve the above curve are called \emph{optimal} regenerating codes. This tradeoff curve is equipped with two extreme points, namely, the minimum bandwidth regenerating (MBR) point and the minimum storage regenerating (MSR) point, respectively representing optimal regenerating codes with the least repair bandwidth and ones with the least per node storage. The corresponding parameters of the two points are determined by
\begin{equation}
\left\{\begin{aligned}
&(\alpha_{\mathbf{MSR}},\beta_{\mathbf{MSR}})=(\frac{B}{k}, \frac{B}{k(d-k+1)})\\
&(\alpha_{\mathbf{MBR}},\beta_{\mathbf{MBR}})=(\frac{2dB}{k(2d-k+1)}, \frac{2B}{k(2d-k+1)}).
\end{aligned}\right.
\end{equation}

There are three repair models considered in literature: functional repair, exact repair and exact repair of systematic nodes \cite{Re:A.G.Dimakis1}. Compared to the functional repair, the exact repair is highly favorable, since exact repair can restore the exact replicas of the lost data \cite{Re:C.Huang}. In the scenario of exact repair, Shah et al. in \cite{Re:Shah.N.B} demonstrate that it is not achievable for most interior points on the tradeoff curve. For those possibly achievable interior points, there are few constructions of codes \cite{Re:C.Tian,Re:T.Ernvall}. Hereinafter, we restrict our focus to exact repair of linear MSR codes.

\subsection{Minimum Storage Regenerating Codes}

Apart from the concept of regenerating codes, array codes \cite{Re:Blaum} are extensively studied over several years and have been widely deployed in DSSs \cite{Re:K.Rashmi,Re:C.Suh,Re:Y.Wu1,Re:K.V.Rashmi,Re:Rashmi,Re:N.B.Shah,Re:Z.Wang,Re:D.S.Papailiopoulos,Re:V.R,Re:I.Tamo,Re:Z.Wang1,Re:Z.Wang2,Re:V.R.Cadambe,Re:V.R.Cadambe1,Re:G.K.Agarwal,Re:B.S,Re:Y.S.Han,Re:J.Li,Re:A.S,Re:N.Raviv,Re:M.Ye,Re:Sasidharan}.
An $\{n,k,l\}$-array code is built from an $(l\times n)$ matrix, where each column is the codeword stored in the $i$-th node. The notation $l$ known as the \emph{sub-packetization size} \cite{Re:Cadambe} basically represents the same meaning with the parameter $\alpha$ of linear regenerating codes \cite{Kun}. In this sense, as for optimal-repair MDS array codes with $\{d=n-1\}$ that reside in the MSR scenario, the corresponding parameter notations can be rewritten as $\{n,k,\alpha=l,\beta=\frac{l}{r},B=kl\}$ where $r=n-k$. Additionally, any optimal-repair systematic MDS array code obviously belongs to the scenario of MSR codes, since the systematic structure of a code signifies its property of minimum storage.

From the perspective of coding theory and technique for MSR codes, a lot of progresses have been made. Collectively, MSR codes are categorized into \emph{scalar} MSR codes with $\{\beta=1\}$ \cite{Re:K.Rashmi,Re:C.Suh,Re:Y.Wu1,Re:K.V.Rashmi,Re:Rashmi,Re:N.B.Shah} and \emph{vector} MSR codes with $\{\beta>1\}$ \cite{Re:Z.Wang,Re:D.S.Papailiopoulos,Re:V.R,Re:I.Tamo,Re:Z.Wang1,Re:Z.Wang2,Re:V.R.Cadambe,Re:V.R.Cadambe1,Re:G.K.Agarwal,Re:B.S,Re:Y.S.Han,Re:J.Li,Re:A.S,Re:N.Raviv,Re:M.Ye,Re:Sasidharan}.
When $d=n-1$, scalar MSR codes are formed by $\{l=r\}$ and the majority of these vector MSR codes are provided with $\{\beta=(n-k)^{x}\leftrightarrow l=r^{x+1}, x\geq1\}$. Most of current constructions are established on the technique of \emph{interference alignment} \cite{Re:N.B.Shah} that is necessary for constructing linear MSR codes. Moreover, it is verified in \cite{Re:N.B.Shah} that there only exist scalar linear MSR codes within the low rate regime $\{\frac{k}{n}\leq\frac{1}{2}\}$. The known product-matrix-based MSR code \cite{Re:K.Rashmi} is scalar for $\{2k-2\leq d\leq n-1\}$. Another design for scalar MSR codes with $\{d=n-1\geq 2k-1\}$ is given in \cite{Re:C.Suh} that entirely stems from interference alignment. In the high rate regime $\{\frac{k}{n}>\frac{1}{2}\}$, vector MSR codes can be leveraged to allow arbitrarily high rates to be attained. However, many of them allow optimal repair of only systematic part \cite{Re:V.R,Re:I.Tamo,Re:Z.Wang1,Re:Z.Wang2,Re:V.R.Cadambe,Re:V.R.Cadambe1,Re:G.K.Agarwal,Re:B.S,Re:Y.S.Han,Re:J.Li,Re:A.S,Re:N.Raviv} that are also referred to as systematic-repair\footnote{The two concepts of systematic MSR codes and systematic-repair MSR codes have generated differences in the definition of MSR codes. As described in \cite{Re:A.Dimakis}, the formal description of an MSR code requires all nodes be optimally repairable. But, many references in the literature also term a code as an MSR code even if it allows optimal repair of systematic nodes only. In this case, we will distinguish between two types of MSR codes by calling them as systematic-repair MSR codes and systematic MSR codes with optimal repair of all nodes (or just systematic MSR codes for simplicity) respectively. When the context is clear or the distinction is unnecessary, they both are named by MSR codes.} MSR codes for brevity, such as Permutation codes \cite{Re:V.R} and Zigzag codes \cite{Re:I.Tamo}. Vector MSR codes permitting optimal repair of parity nodes as well are proposed in \cite{Re:Z.Wang,Re:D.S.Papailiopoulos,Re:M.Ye}, where the code given in \cite{Re:Z.Wang} is a variant of Zigzag code. On the other hand, Cadambe et al. in \cite{Re:Cadambe} prove that with $l$ (or $\alpha$) scaling to infinity, there exist high-rate MSR codes with parameter set $\{n,k,d\}$ of any value. Goparaju et al. in \cite{Re:S.G} provide an non-explicit construction of systematic-repair MSR codes for all $\{k\leq d\leq n-1\}$ that can be used to meet the high-rate requirement. For $d<n-1$, authors in \cite{Re:Sasidharan} present a special $\{l=\frac{1}{2}\cdot r^{\frac{n}{r}}\}$ vector MSR code with an explicit coupled-layer construction that optimally repairs all nodes, which in fact is built on the design of \cite{Re:B.S} and also can be adjusted to have $\frac{k}{n}$ as close to $1$ as desirable high rate.

\subsection{Existing Upper Bounds on Systematic-Length}

Centering on the case of systematic-repair MSR codes parameterized by $\{n=k+r,k,l,d=n-1\}$, we revisit the problem considered in \cite{Re:I.Tamo1,Re:S.Goparaju1} ``the relationship between the node storage capacity $l$ and the number of systematic nodes $k$ for some constant $r$". To be precise, the concern is as follows: what is the longest systematic-length $k$ for which there exists a linear systematic-repair MSR code, for a given sub-packetization level $l$ and a given number of parity nodes $r$?

Tamo et al. in \cite{Re:I.Tamo1} initially bring forward this question. They first consider systematic-repair MSR codes with two properties of particular interest, i.e., \emph{optimal access} and \emph{optimal update}, where an optimal-access code transmits only the symbols it accesses and an optimal-update one updates exactly once in each parity node. Restricted to the condition of optimal update, they demonstrate that the code with optimal access has the same longest systematic-length $k$ as the one with optimal bandwidth only, where $k_{max}=\log_rl$. For instance, Zigzag codes \cite{Re:I.Tamo} are of optimal-update that can attain this value. Without this restriction, they further show that $k=r\log_rl$ is the longest systematic-length for a systematic-repair MSR code with optimal access. The scheme for optimal-access and systematic-repair MSR code that achieves this optimal value of systematic-length is initially proposed in \cite{Re:V.R.Cadambe1}. Recently, authors in \cite{Re:N.Raviv} also present an explicit construction for systematic-repair MSR codes with optimal access, where $l=r^{\frac{k}{r}}$ reaches the optimal value. Another construction of systematic-repair MSR codes with optimal access given in \cite{Re:M.Ye} requires that $l=r^{\lceil\frac{n}{r}\rceil}=r^{\lceil\frac{k}{r}\rceil+1}$, which however is unattainable and differs from $l_{min}=r^{\frac{k}{r}}$ by $r^2$ approximately.

In regard to the general situation (namely the case of non-optimal update as well as non-optimal access), Tamo et al. only provide a loose upper bound
\begin{equation}
k\leq l{{l}\choose{\frac{l}{r}}},
\end{equation}
and conjecture that $k$ is of the order of $\log_rl$. In \cite{Re:Z.Wang1} and its extended version \cite{Re:Z.Wang2}, authors design a general systematic-repair MSR code, wherein $l=r^{\frac{k}{r+1}}$ or $k=(r+1)\log_rl$ that is consistent with this conjecture. Nevertheless, it still remains open to explore better upper bounds on systematic-length $k$ and simultaneously construct tighter systematic-repair MSR codes for the general case \cite{Re:Z.Wang2}.

Aiming at the above general situation, Goparaju et al. in \cite{Re:S.Goparaju1} transform it into a linear algebraic problem involving a set of linear subspaces and operators, where \emph{repair subspaces}\footnote{The formal description of repair subspace stems from \cite{Re:S.Goparaju1}, that can be also clarified from the phrases on top of Note \ref{note1} and Equation (\ref{interference}) in this paper. Simply put, repair subspaces correspond to a set of matrices over $\mathbb{F}^{\frac{l}{r}\times l}$ that, in conjunction with stored data vectors, are used to generate repair data symbols.} are the main research target herein and the operators in fact are a series of invertible square matrices. Through operating these square matrices on repair subspaces for purpose of satisfying the conditions of subspace properties, they sequentially analyze and demonstrate the linear independence among these operators, where the subspace properties actually are established on the technique of interference alignment as shown in \cite{Re:I.Tamo1}. In particular, they mainly work on the case of $r=2$ based upon the subspace properties for any distinct $i,j\in[1,k]$ as follows:
\begin{equation}
\left\{\begin{aligned}
&\mathbf{S}_{i}\simeq\mathbf{S}_{i}\mathbf{\Phi_j}\\
&\mathbf{S}_{i}\mathbf{\Phi_i}\oplus \mathbf{S}_{i} \simeq \mathbb{F}^l,
\end{aligned}\right.
\end{equation}
where $\mathbf{S}_{i}$ is the repair subspace of the $i$-th systematic node with dimension $\frac{l}{2}$, $\mathbf{\Phi_i}$ is the corresponding operator of order $l$, $``\oplus"$ represents the direct sum and $``\simeq"$ stands for a symbol for identical space\footnote{In subsequent section, we give the extended subspace properties for the case of arbitrary number of parity nodes $r$. With these extended properties, we will then present two improved upper bounds on $k$.}. Briefly speaking, they take advantage of the restricted geometric properties of subspaces and operators, by which they construct linearly independent operators of special features and finally derive three new upper bounds on $k$ owing to the limited dimension of matrices space $\mathbb{F}^{l\times l}$ that is equal to $l^2$.

They first present a simple systematic-length bound $k\leq l^2$ that does not depend on the number of parity nodes $r$. For the special case of $2$ parity nodes, they prove that $k\leq 4l+1$. Subsequently, employing the method of \emph{partition}\footnote{The method of partition is introduced from \cite{Re:S.Goparaju1} that basically stands for dividing the systematic part $[1,k]$ in a certain way, which can be technically described as follows. To be specific, we denote the size of each partition by the notation $\lambda$, i.e., the smallest integer value such that there exist $\lambda$ number of systematic nodes $\{i_1,\cdots,i_{\lambda}\}$ satisfying $\biguplus_{\nu=1}^{\lambda} \mathbf{S}_{i_\nu}\simeq \mathbb{F}^l$ $\big($or $\dim(\biguplus_{\nu=1}^{\lambda} \mathbf{S}_{i_\nu})=l$$\big)$ for $i_\nu\in[1,k]$, where $\mathbf{S}_{i_\nu}$ is the repair subspace of the $i_\nu$-th systematic node and $``\biguplus"$ represents the sum of subspaces as defined in Note \ref{note1}. Thereupon, it signifies that every collection of $\lambda$ repair subspaces among $\{\mathbf{S}_1,\ldots,\mathbf{S}_k\}$ can span the whole space, i.e., repair subspaces within each partition sized by $\lambda$ span the entire space $\mathbb{F}^l$.} for systematic part $[1,k]$ where repair subspaces within each partition span the entire space $\mathbb{F}^l$, they design $2^{\frac{k}{\lambda}}$ linearly independent matrices and thereby derive
\begin{equation}
k\leq  2\lambda\log_2l,
\end{equation}
where $\lambda$ is the size of each partition. Applying the geometric analysis of repair subspaces, they prove that $\lambda\leq\left \lfloor\log_{\frac{r}{r-1}}l\right \rfloor+1$, which means that if the partition contains $\big(\left \lfloor\log_{\frac{r}{r-1}}l\right \rfloor+1\big)$ number of repair subspaces, they necessarily span the whole space $\mathbb{F}^l$. As a matter of fact, the second case is included herein, since it is easy to verify that $2\log_2l\big(\left \lfloor\log_{2}l\right \rfloor+1\big)< 4l+1$.

\begin{exa}\label{exa}
Consider a tiny storage disk of size that equals with $2^{8}$ bits, i.e., $l=2^{8}$. As for the case of $r=16$ parity nodes, from the latest upper bound on $k$, we have that $k\leq 2\big(\log_22^8\big)\big(\log_{\frac{16}{15}}2^8+1\big)\backsimeq1391$. With our new upper bound present in this paper, we will prove that $k$ cannot exceed $348$ (see the next section for details).
\end{exa}

\subsection{Our Contribution}

In this work, we continue to investigate the aforementioned general situation, i.e., for simplicity, seek tighter upper bounds on systematic-length $k$ for general systematic-repair MSR codes with a given $l$ and $r$. We study by means of further exploiting the geometric analysis of linear subspaces and operators applied in \cite{Re:S.Goparaju1} for the case of arbitrary number of parity nodes $r$.

By virtue of the basic principle of interference alignment and some equivalent translations, we begin with formalizing two crucial derivative extended properties of subspaces and operators: (i) it is closed for operations on repair subspaces by the addition and multiplication of encoding matrices, and (ii) any vectors coming from different linear subspaces that mutually intersect trivially are linearly independent. Using these two properties, we sequentially find the linear independence among the existing encoding matrices and further proceed to create new linearly independent matrices, which both naturally bound the systematic-length $k$ due to the dimensional restriction of square matrices.

First of all, we prove that all the encoding matrices are linearly independent, with which we obtain a simple quadratic bound
\begin{equation}
k\leq \frac{l^2-1}{r-1}.
\end{equation}
This simple bound implies that $k=r+2$ is the longest systematic-length in the scalar scenario.

Secondly, we follow the idea of partition employed in \cite{Re:S.Goparaju1}, where we construct new matrix by ``\emph{adding}" matrices of the same row within each partition. The reason for adding matrices instead of multiplying them as in \cite{Re:S.Goparaju1} is that matrix's addition possesses the commutative law while multiplication of matrix is non-commutative. Afterwards, through multiplying these constructed matrices from different partitions in sequence, we further strike out $r^{\frac{k}{\lambda}}$ new matrices, where $\lambda$ is the size of a partition. We demonstrate that these $r^{\frac{k}{\lambda}}$ new matrices are all non-zero matrices and linearly independent as well, from which we derive an $r$-based-log upper bound
\begin{equation}
k\leq 2\lambda\log_rl.
\end{equation}

Finally, we consider estimating the partition size $\lambda$, where Goparaju et al. in \cite{Re:S.Goparaju1} show that $\lambda\leq \left \lfloor\log_{\frac{r}{r-1}}l\right \rfloor+1$. Leveraging our previous work in \cite{Kun}, we introduce an explicit result as follows
\begin{equation}
\lambda:\left\{\begin{array}{ll}
=r, &\textrm{when} \quad t=r;\\
\leq t+1+\left \lfloor\log_{\frac{r}{r-1}}\frac{(r-t)l}{r}\right \rfloor, &\textrm{when} \quad 1\leq t<r,
\end{array}\right.
\end{equation}
where $t=\lceil\frac{r^2}{l}\rceil$, i.e., $t$ is an integer such that $\frac{r}{t}\leq \frac{l}{r}<\frac{r}{t-1}$ or $\frac{r^2}{l}\leq t<\frac{r^2}{l}+1$. It indicates that $\lambda=r$ if $t=r$ and $\lambda\leq\left \lfloor\log_{\frac{r}{r-1}}l\right \rfloor+1$ for $t=1$, which extend the corresponding result in \cite{Re:S.Goparaju1}. Furthermore, it can be verified that $r\leq \lambda\leq \left \lfloor\log_{\frac{r}{r-1}}l\right \rfloor+1$. To be more visible, when setting $l=2^8$ and $r=16$ as in Example \ref{exa}, we have that $t=\lceil\frac{16^2}{2^8}\rceil=1$ and therewith our $r$-based-log upper bound leads to that $k\leq 2\big(\log_{16}2^8\big)\big(\log_{\frac{16}{15}}2^8+1\big)\backsimeq348$.

\vspace{0.2cm}

Table \ref{Tab:Comparision} illustrates the research progresses on systematic-length $k$ for linear systematic-repair MSR codes\footnote{In subsequent description, construction of linear systematic-repair MSR codes has to satisfy the condition required for the interference alignment, which however aims at designing MSR codes with only optimally repairing systematic nodes. Since interference alignment does not apply to systematic MSR codes (with optimally repairing all nodes), it is evident that upper bounds on $k$ for systematic MSR codes should be tighter than the one for systematic-repair MSR codes. For the sake of clarity, we stress again that linear systematic-repair MSR codes are our research object in this paper.} including upper bounds and the scheme \footnote{Although the construction scheme by Wang et al. \cite{Re:Z.Wang2} is a general case with the longest systematic-length $k=(r+1)\log_rl$, it is obviously shorter than our new upper bound $2\lambda\log_rl$ where $r\leq \lambda\leq \left \lfloor\log_{\frac{r}{r-1}}l\right \rfloor+1$. That is to say, the longest one so far given in \cite{Re:Z.Wang2} is still not achievable.} with the longest known systematic-length \cite{Re:Z.Wang2}.

\begin{table*}[ht]
 \tabcolsep 0pt
 \begin{center}
\def\temptablewidth{0.93\textwidth}
\caption{Progresses on Systematic-Length $k$ for Linear Systematic-Repair MSR Codes.}\label{Tab:Comparision}
{\rule{\temptablewidth}{1pt}}
\begin{tabular}{|c|c|}
 ~~~~~~~Citation~      & ~Corresponding Results~          \\ \hline
  ~Tamo et al. \cite{Re:I.Tamo1}~      & $k\leq l{{l}\choose{\frac{l}{r}}}$           \\ \hline
  ~Goparaju et al. \cite{Re:S.Goparaju1}~ & $ k\leq l^2$ $\textbf{\&}$ $k\leq  2\lambda\log_2l$, where $\lambda\leq\left \lfloor\log_{\frac{r}{r-1}}l\right \rfloor+1$    \\ \hline
  ~~~~Wang et al.\cite{Re:Z.Wang2}~~&     the longest known scheme with systematic-length $k=(r+1)\log_rl$     \\ \hline
 ~~This Paper & ~ $k\leq \frac{l^2-1}{r-1}$ $\textbf{\&}$ $k\leq 2\lambda\log_rl$, where $\lambda:\left\{\begin{array}{ll}
=r, &\textrm{when} \quad t=r;\\
\leq t+1+\left \lfloor\log_{\frac{r}{r-1}}\frac{(r-t)l}{r}\right \rfloor, &\textrm{when} \quad 1\leq t<r.
\end{array}\right.$         \\ \hline
  \end{tabular}
\end{center}
\end{table*}

\subsection{Organization}

Section 2 gives preliminaries including the basic system setting, the technique of interference alignment and two useful properties related to subspaces. Section 3 presents the detailed proof of our new results for upper bounds on systematic-length. Section 4 concludes this paper.

\section{PRELIMINARIES}

In this section, we first describe the system setting about the basic construction of linear systematic-repair MSR codes. Then, the technique of interference alignment is introduced, where some equivalent transformations are applied. Last, we formalize two properties of subspaces and operators, which will be used to find new linearly independent matrices.

\subsection{System Setting}

Consider a linear systematic-repair MSR code with $\{n=k+r,k,l\}$, where each of $n$ nodes stores $l$ data symbols over a finite field $\mathbb{F}$ such that any $k$ nodes are sufficient to recover the $k\cdot l$ original data symbols. In light of its systematic feature, we let $(W_1,\cdots,W_k)$ denote these original data symbols stored in the $k$ systematic nodes. Each of the remaining $r$ nodes, as a parity node, stores a linear combination of $(W_1,\cdots,W_k)$. Formally speaking, the $l$ encoded symbols of $W_{k+u}$ stored in parity node $u$ for each $u\in[1,r]$, can be expressed as
\begin{equation}\label{encode}
W_{k+u}=\sum_{j=1}^{k}\mathbf{C}_{u,j}W_j,
\end{equation}
where $\mathbf{C}_{u,j}$ is a square encoding matrix of order $l$ corresponding to the parity node $u$ and the systematic node $j\in[1,k]$. Hence, any linear systematic-repair MSR code can be uniquely represented as follows
\begin{equation}\label{encode}
\mathbb{C}=(\mathbf{C}_{u,j})_{\{u\in[1,r],j\in[1,k]\}}=\left[
  \begin{array}{ccccc}
    \mathbf{C}_{1,1}               &\cdots  & \mathbf{C}_{1,k}        \\
    \vdots                         &\ddots  &\vdots    \\
    \mathbf{C}_{r,1}               &\cdots  &\mathbf{C}_{r,k} \\
  \end{array}
\right].
\end{equation}
The MDS property of systematic-repair MSR codes requires that any $r$ node failures can be recovered, which is equivalent to that any block submatrix of (\ref{encode}) has to be invertible. Consequently, each encoding matrix $\mathbf{C}_{u,j}$ also must be invertible.

In order to optimally repair any systematic node $i\in[1,k]$, the remaining $n-1$ nodes are required to send a fraction $\frac{l}{r}$ of their data stored, i.e., each helper node $\{\nu\neq i|\nu\in[1,n]\}$ sends $\frac{l}{r}$ repair data symbols represented by $\mathbf{S}_{i,\nu}W_\nu$, where $\mathbf{S}_{i,\nu}$ is a matrix over $\mathbb{F}^{\frac{l}{r}\times l}$. To this end, each repair data made of $\frac{l}{r}$ data symbols sent from helper node $\nu$ to systematic node $i$ can be regarded as the projection of $W_\nu$ onto a subspace of dimension $\frac{l}{r}$, which corresponds to the subspace spanned by the rows of $\mathbf{S}_{i,\nu}$ that is sequentially called repair subspace.

\begin{no}\label{note1}
To avoid confusion of the subsequent discussion, we let $``\sum"$ denote the normal sum of matrices, $``\biguplus"$ be the sum of subspaces (spanned by the rows of the corresponding matrices) and $``\bigoplus"$ represent the direct sum of certain subspaces that mutually intersect trivially. Besides, $\mathbf{S}_{i}$ and $\mathbf{S}_{i}\mathbf{C}_{u,j}$ can represent either a $(\frac{l}{r}\times l)$ matrix or a subspace of dimension $\frac{l}{r}$. Furthermore, we let $``0"$ denote the scalar zero, zero-matrix or zero-subspace, which will be clear from the context.
\end{no}

\subsection{Interference Alignment}

As shown in \cite{Re:Z.Wang1}, a linear systematic-repair MSR code has to meet the requirement of interference alignment. Optimal repair of a systematic node $i$ is possible if and only if there exist $\big\{\mathbf{S}_{i,\nu}|\nu\neq i,\nu\in[1,n]\big\}$ satisfying, for any $u\in[1,r]$ and $j\in[1,k]$ with $j\neq i$,
\begin{equation}\label{equ}
\left\{\begin{aligned}
&\mathbf{S}_{i,j}\simeq\mathbf{S}_{i,k+u}\mathbf{C}_{u,j}\\
&\biguplus_{u=1}^{r}\mathbf{S}_{i,k+u}\mathbf{C}_{u,i} \simeq \mathbb{F}^l,
\end{aligned}\right.
\end{equation}
where $``\simeq"$ means that the subspaces spanned by the rows of the corresponding matrices are identical. Besides, the sum of the subspaces $``\biguplus"$ must be a direct sum $``\bigoplus"$, since the dimension of each subspace $\{\mathbf{S}_{i,\nu}|\nu\neq i,\nu\in[1,n]\}$ is $\frac{l}{r}$ and each matrix $\mathbf{C}_{u,i}$ is invertible.

\vspace{0.2cm}

As explained in \cite{Re:I.Tamo1}, each encoding matrix for one of the parity nodes in a linear systematic-repair MSR code can be assumed to be the identity matrix $\mathbf{I}$, i.e., $\mathbf{C}_{1,i}=\mathbf{I}$ for any $i\in[1,k]$. Thereafter, the authors prove that if there exists a $\{k+r+1,k+1,l\}$ linear systematic-repair MSR code, then there also exists a $\{k+r,k,l\}$ linear systematic-repair MSR code where the repair subspaces are independent of the helper nodes. It indicates that, for a given $i\in[1,k]$, $\mathbf{S}_{i,\nu}$ stays unchanged for any $\nu\in[1,n]$ and $\nu\neq i$, which can be replaced by $\mathbf{S}_{i}$ accordingly.

\vspace{0.2cm}

Henceforth, the subspace conditions (\ref{equ}) almost can be equivalently transformed into:
\begin{equation}\label{interference}
\left\{\begin{aligned}
&\mathbf{S}_{i}\simeq\mathbf{S}_{i}\mathbf{C}_{u,j}\\
&\bigoplus_{u=2}^{r}\mathbf{S}_{i}\mathbf{C}_{u,i}\oplus \mathbf{S}_{i} \simeq \mathbb{F}^l,
\end{aligned}\right.
\end{equation}
where $i,j\in[1,k]$ with $j\neq i$ and $\mathbf{S}_{i}$ is the repair matrix (or viewed as the independent repair subspace) for systematic node $i$.

\subsection{Extended Properties of Subspaces}

From the subspace conditions (\ref{interference}), two derivative extended properties are formally presented as follows, which can be counted as a generalization of some results derived in \cite{Re:S.Goparaju1}.

\begin{lemma}\label{premise1}
Given an $\{n=k+r,k,l\}$ linear systematic-repair MSR code with independent repair subspaces, for any $u_1,u_2\in[1,r]$ and any $i,j_1,j_2\in[1,k]$ such that $j_1\neq i$ and $j_2\neq i$, it must be that
\begin{equation}
\left\{\begin{aligned}
&\mathbf{S}_{i}(\mathbf{C}_{u_1,j_1}+\mathbf{C}_{u_2,j_2})\preceq\mathbf{S}_{i};\\
&\mathbf{S}_{i}\mathbf{C}_{u_1,j_1}\mathbf{C}_{u_2,j_2}\preceq\mathbf{S}_{i},
\end{aligned}\right.
\end{equation}
where $``\preceq"$ represents a symbol of inclusion for subspaces ($x\preceq y$ means that the subspace $x$ lies in the subspace $y$).
\end{lemma}

\begin{proof}

According to the first item in (\ref{interference}), $``\mathbf{S}_{i}\simeq\mathbf{S}_{i}\mathbf{C}_{u,j}"$ signifies that the subspace spanned by the $\frac{l}{r}$ row vectors of $\mathbf{S}_{i}$ is same as the one spanned by the $\frac{l}{r}$ row vectors of $\mathbf{S}_{i}\mathbf{C}_{u,j}$. Equivalently speaking, there exists an invertible $(\frac{l}{r}\times \frac{l}{r})$ matrix $\mathbf{P}_{u,i,j}$ such that $\mathbf{S}_{i}\mathbf{C}_{u,j}=\mathbf{P}_{u,i,j}\mathbf{S}_{i}$.

\vspace{0.2cm}

Thereby, we have
\begin{equation}
\left\{\begin{aligned}
&\mathbf{S}_{i}(\mathbf{C}_{u_1,j_1}+\mathbf{C}_{u_2,j_2})=(\mathbf{P}_{u_1,i,j_1}+\mathbf{P}_{u_2,i,j_2})\mathbf{S}_{i};\\
&\mathbf{S}_{i}\mathbf{C}_{u_1,j_1}\mathbf{C}_{u_2,j_2}=\mathbf{P}_{u_1,i,j_1}\mathbf{P}_{u_2,i,j_2}\mathbf{S}_{i}.
\end{aligned}\right.
\end{equation}
From the above equation, it is clear that $\mathbf{S}_{i}\mathbf{C}_{u_1,j_1}\mathbf{C}_{u_2,j_2}\simeq\mathbf{S}_{i}$, since $\mathbf{P}_{u_1,i,j_1}\mathbf{P}_{u_2,i,j_2}$ is an invertible matrix. Moreover, because $\mathbf{P}_{u_1,i,j_1}+\mathbf{P}_{u_2,i,j_2}$ also is a $(\frac{l}{r}\times \frac{l}{r})$ matrix, it apparently leads to that $(\mathbf{P}_{u_1,i,j_1}+\mathbf{P}_{u_2,i,j_2})\mathbf{S}_{i}\preceq\mathbf{S}_{i}$.$\hfill\blacksquare$

\end{proof}

\begin{rem}
Lemma \ref{premise1} indicates the closure of operations on $\mathbf{S}_{i}$ by the addition and multiplication of $\mathbf{C}_{u,j}$ for any $j\neq i$ and $u\in[2,r]$. Consequently, any subspace generated from operating on $\mathbf{S}_{i}$ by a matrix made up of arbitrary times of addition and multiplication of matrices $\mathbf{C}_{u,j}$ still lies in $\mathbf{S}_{i}$, for $j\neq i$.
\end{rem}

\begin{lemma}\label{premise2}
Given an $\{n=k+r,k,l\}$ linear systematic-repair MSR code with independent repair subspaces, for any $i\in[1,k]$, if there exist $r$ matrices $\{\mathbf{\Theta}_1,\cdots,\mathbf{\Theta}_r\}$ of size $(\frac{l}{r}\times l)$ such that $\mathbf{\Theta}_u\preceq\mathbf{S}_{i}\mathbf{C}_{u,i}$ for each $u\in[1,r]$ and $\sum_{u=1}^r\mathbf{\Theta}_u=0$, then it must be that $$\mathbf{\Theta}_u=0.$$
\end{lemma}

\begin{proof}
According to the second item in (\ref{interference}), $``\bigoplus_{u=2}^{r}\mathbf{S}_{i}\mathbf{C}_{u,i}\oplus \mathbf{S}_{i} \simeq \mathbb{F}^l"$ indicates that all the $r$ subspaces $\{\mathbf{S}_{i}\mathbf{C}_{u,i}|u\in[1,r]\}$ where $\mathbf{C}_{1,i}=\mathbf{I}$, exactly span the entire space $\mathbb{F}^l$. That is to say, the total $l$ row vectors of $\{\mathbf{S}_{i}\mathbf{C}_{u,i}|u\in[1,r]\}$ are linearly independent. Thus, it leads to that any $r$ non-zero vectors coming from $r$ different subspaces $\mathbf{S}_{i}\mathbf{C}_{u,i}$ for $u\in[1,r]$ are also linearly independent.

\vspace{0.2cm}

Hereby, we assume there are non-zero matrices $\mathbf{\Theta}_u$ for some $u\in[1,r]$. However, the condition $\sum_{u=1}^r\mathbf{\Theta}_u=0$ implies that all non-zero vectors of the same row within these non-zero matrices $\mathbf{\Theta}_u$ are not linearly independent. Since these non-zero vectors lie in the corresponding subspaces $\mathbf{S}_{i}\mathbf{C}_{u,i}$, contradiction arises.$\hfill\blacksquare$

\end{proof}

\begin{rem}
Lemma \ref{premise2} in fact is a derivative property of direct sum of linear subspaces, which means that those vectors from different linear subspaces that mutually intersect trivially are linearly independent.
\end{rem}

Subsequently, we use Lemma \ref{premise1} and Lemma \ref{premise2} to find new linearly independent matrices, which will naturally offer the upper bounds on the systematic-length $k$.

\section{IMPROVED SYSTEMATIC-LENGTH UPPER BOUNDS}

In this section, we will present two new upper bounds on systematic-length. The first one is a simple quadratic bound, which stems from the linear independence of encoding matrices. The second one is an $r$-based-log upper bound, which is based on the partition of systematic nodes $[1,k]$, where repair subspaces within each standard partition span the entire space $\mathbb{F}^l$. Furthermore, we estimate the size of the standard partition and introduce an explicit result depending on the value of $\frac{r^2}{l}$, which can be integrated into the $r$-based-log bound.

\subsection{A Simple Quadratic Bound}
\begin{theorem}\label{scalar}
For any $\{n=k+r,k,l\}$ linear systematic-repair MSR code with independent repair subspaces, the following upper bound holds
\begin{equation}
k\leq \frac{l^2-1}{r-1}.
\end{equation}

\end{theorem}

\begin{proof}

We proceed to show that the $k(r-1)+1$ number of matrices $\Big\{\mathbf{I},\mathbf{C}_{u,i}|u\in[2,r],i\in[1,k]\Big\}$ are linearly independent.

\vspace{0.2cm}

Assume there exist $k(r-1)+1$ coefficients $a_{u,i}$ and $b$ in $\mathbb{F}$ such that
\begin{equation}\label{nondependent}
\sum_{i=1}^{k}\sum_{u=2}^{r} a_{u,i}\mathbf{C}_{u,i}+b\mathbf{I}=0,
\end{equation}
where $``0"$ represents the zero-matrix.

\vspace{0.2cm}

Then, operating the above equation on $\mathbf{S}_{i}$ for each $i\in [1,k]$, we have
\begin{equation}\label{independent matrix}
\sum_{u=2}^{r} a_{u,i}\mathbf{S}_{i}\mathbf{C}_{u,i}+(\sum_{j\neq i}\sum_{u=2}^{r}a_{u,j}\mathbf{S}_{i}\mathbf{C}_{u,j}+b\mathbf{S}_{i})=0,
\end{equation}
where Lemma \ref{premise1} leads to that $(\sum_{j\neq i}\sum_{u=2}^{r}a_{u,j}\mathbf{S}_{i}\mathbf{C}_{u,j}+b\mathbf{S}_{i})\preceq\mathbf{S}_{i}$. With Lemma \ref{premise2}, we know that, for any $u\in[2,r]$ and $i\in[1,k]$,
$$a_{u,i}=0.$$
Combining with equation (\ref{nondependent}), we further derive
$$b=0,$$
which exactly means that the $k(r-1)+1$ matrices $\Big\{\mathbf{I},\mathbf{C}_{u,i}|u\in[2,r],i\in[1,k]\Big\}$ are linearly independent\footnote{As for an arbitrary MDS code, the linear independence among the matrices $\Big\{\mathbf{I},\mathbf{C}_{u,i}|u\in[2,r],i\in[1,k]\Big\}$ cannot always hold. The reason is that unlike systematic-repair MSR code, an arbitrary MDS code is not necessarily equipped with the property of interference alignment and as a consequence cannot meet the conditions required in Lemma \ref{premise1} and Lemma \ref{premise2}.}.

\vspace{0.2cm}

Since they all lie in the $l^2$-dimensional space of matrices $\mathbb{F}^{l\times l}$, we have
\begin{equation}
k(r-1)+1\leq l^2\Leftrightarrow k\leq \frac{l^2-1}{r-1}.
\end{equation}$\hfill\blacksquare$
\end{proof}

\begin{rem}
In \cite{Re:S.Goparaju1}, Goparaju et al. obtain that $k\leq l^2$ under the case of two parity nodes. Here, we look into the case of arbitrary number of parity nodes $r$ and derive that $k\leq \frac{l^2-1}{r-1}$.

In addition, we find that when $l=r$, it has to be that $k\leq r+1$. Without the premise that the repair subspaces are independent of the helper nodes, we further have $k\leq r+2$ that is exactly consistent with the corresponding result in \cite{Re:N.B.Shah}, i.e., scalar linear systematic-repair MSR codes only exist when $d\geq 2k-3$ for $d=n-1$. That is to say, in the scalar linear systematic-repair MSR scenario, $k=r+2$ is the longest systematic-length.
\end{rem}

\subsection{An $r$-Based-Log Bound}
\begin{theorem}\label{log}
For any $\{n=k+r,k,l\}$ linear systematic-repair MSR code with independent repair subspaces, the following upper bound holds
\begin{equation}
k\leq 2\lambda\log_rl,
\end{equation}
where $\lambda$ as defined in Footnote 4 represents the size of each partition.
\end{theorem}

\begin{proof}
We follow the idea of partition considered in \cite{Re:S.Goparaju1}, while we employ the addition of matrices within each partition instead of matrices' multiplication therein. The basic difference is that matrix's addition has the commutative law while multiplication of matrix is non-commutative. We analyze two situations as follows.
\vspace{0.2cm}

\textbf{1. The situation when $\lambda$ divides $k$.}
\vspace{0.2cm}

Divide the indices of systematic nodes $[1,k]$ into $\frac{k}{\lambda}$ disjoint partitions $\{\mathcal{X}_1,\mathcal{X}_2,\cdots,\mathcal{X}_{\frac{k}{\lambda}}\}$, where $\mathcal{X}_i=[1+(i-1)\lambda,\cdots,i\lambda]$. Accordingly, partition $\{\mathbf{S}_1,\cdots,\mathbf{S}_{k}\}$ and $\Big\{\mathbf{I},\mathbf{C}_{u,i}|u\in[2,r],i\in[1,k]\Big\}$ as follows
\begin{equation}
\left[
  \begin{array}{ccc;{3pt/3pt}ccc;{3pt/3pt}c;{3pt/3pt}ccc}
   \quad \mathbf{S}_1   &\cdots  & \mathbf{S}_{\lambda}\quad  & \quad \mathbf{S}_{\lambda+1} & \cdots  & \mathbf{S}_{2\lambda} \quad & \quad \cdots \quad & \quad\mathbf{S}_{k-\lambda+1}  &\cdots &\mathbf{S}_{k}\quad     \\ \hline
  \quad  \mathbf{I} & \cdots & \mathbf{I} \quad& \quad \mathbf{I} & \cdots & \mathbf{I} \quad &\quad\cdots\quad &\quad\mathbf{I} &\cdots & \mathbf{I}\quad \\ \hdashline[2pt/2pt]
  \quad  \mathbf{C}_{2,1} & \cdots & \mathbf{C}_{2,\lambda} \quad & \quad \mathbf{C}_{2,\lambda+1} & \cdots & \mathbf{C}_{2,2\lambda} \quad &\quad\cdots\quad &\quad\mathbf{C}_{2,k-\lambda+1} &\cdots & \mathbf{C}_{2,k}\quad \\\hdashline[2pt/2pt]
   \quad  \vdots & \cdots & \vdots \quad & \quad \vdots & \cdots & \vdots\quad &\quad\cdots\quad &\quad\vdots &\cdots & \vdots \quad\\\hdashline[2pt/2pt]
   \quad \mathbf{C}_{r,1} & \cdots & \mathbf{C}_{r,\lambda} \quad & \quad \mathbf{C}_{r,\lambda+1} & \cdots & \mathbf{C}_{r,2\lambda}\quad &\quad\cdots\quad &\quad\mathbf{C}_{r,k-\lambda+1} &\cdots & \mathbf{C}_{r,k} \quad\\
  \end{array}
\right],
\end{equation}
where all entries of the second row are set to be the identity matrix $\mathbf{I}$.
\vspace{0.2cm}

\textbf{Setting}: For each $u\in [1,r]$ and $i\in[1,\frac{k}{\lambda}]$, define
\begin{equation}
\mathbf{\Gamma}_{i,u}=\left\{\begin{aligned}
&\mathbf{I},~~ if ~~ u=1;\\
&\sum_{j\in \mathcal{X}_i}\mathbf{C}_{u,j}, ~~ if~~ u\in[2,r].
\end{aligned}\right.
\end{equation}

Then, for any $u_i\in[1,r]$ and $i\in[1,\frac{k}{\lambda}]$, we define
$$\mathbf{\Delta}_{u_1u_2\cdots u_{\frac{k}{\lambda}}}=\prod_{i=1}^{\frac{k}{\lambda}}\mathbf{\Gamma}_{i,u_i}.$$
Thus, we totally obtain $r^{\frac{k}{\lambda}}$ square matrices formed as $\mathbf{\Delta}_{u_1u_2\cdots u_{\frac{k}{\lambda}}}$ for $u_i\in[1,r]$, which in fact all are non-zero matrices (It is proved in \textbf{Appendix}). In the following, we further show that these $r^{\frac{k}{\lambda}}$ square matrices are also linearly independent.

\vspace{0.2cm}

\textbf{Induction claim}: Assume for some $s\in [1,\frac{k}{\lambda}]$, the following $r^s$ number of square matrices
$$\left\{\mathbf{\Delta}_{u_1u_2\cdots u_s}=\prod_{i=1}^{s}\mathbf{\Gamma}_{i,u_i}\mid u_i\in[1,r],i\in[1,s]\right\}$$ are linearly independent.

\vspace{0.2cm}
\textbf{Base case}: For $s=1$, we have
\begin{equation}\label{basematrix}
\Big\{\mathbf{\Delta}_{u_1}\mid u_1\in[1,r]\Big\}=\Big\{\mathbf{\Gamma}_{1,u_1}\mid u_1\in[1,r]\Big\}
=\left\{\mathbf{I},\sum_{j=1}^{\lambda}\mathbf{C}_{2,j},\cdots,
\sum_{j=1}^{\lambda}\mathbf{C}_{r,j}\right\}.
\end{equation}
Assume there exist $r$ coefficients $(a_1,\cdots,a_r)$ in $\mathbb{F}$ such that
\begin{equation}\label{base}
a_1\mathbf{I}+a_2\sum_{j=1}^{\lambda}\mathbf{C}_{2,j}+\cdots+a_r\sum_{j=1}^{\lambda}\mathbf{C}_{r,j}=0.
\end{equation}
Operating equation (\ref{base}) on $\mathbf{S}_i$ for any $i\in[1,\lambda]$, we have
\begin{equation}
\sum_{u=2}^r a_u\mathbf{S}_i\mathbf{C}_{u,i}+(\sum_{j\neq i}\sum_{u=2}^r a_u\mathbf{S}_i\mathbf{C}_{u,j}+a_1\mathbf{S}_i)=0,
\end{equation}
where $\sum_{j\neq i}\sum_{u=2}^r a_u\mathbf{S}_i\mathbf{C}_{u,j}+a_1\mathbf{S}_i\preceq\mathbf{S}_i$ following from Lemma \ref{premise1}. According to Lemma \ref{premise2}, we can similarly derive
\begin{equation}
a_1=a_2=\cdots=a_r=0,
\end{equation}
which means the $r$ matrices in (\ref{basematrix}) are linearly independent.

\vspace{0.3cm}
\textbf{Inductive step}: Let the inductive claim hold for some $s$, then it is true for $s+1$. Otherwise, we have
\begin{equation}
\mathbf{\Psi}_1^{(s)}+\mathbf{\Psi}_2^{(s)}\mathbf{\Gamma}_{s+1,2}+\cdots+\mathbf{\Psi}_r^{(s)}\mathbf{\Gamma}_{s+1,r}=0,
\end{equation}
where $\mathbf{\Psi}_u^{(s)}$ for each $u\in[1,r]$ are linear combinations of elements formed as $\mathbf{\Delta}_{u_1u_2\cdots u_s}$\footnote{As assumed in Lemma \ref{premise2}, we let $\mathbf{C}_{u,j}=\mathbf{I}$ when $u=1$ and thereby we can rewrite that $\mathbf{\Gamma}_{i,u}=\sum_{j\in \mathcal{X}_i}\mathbf{C}_{u,j}$ for any $u\in[1,r]$. In this case, $\mathbf{\Psi}_u^{(s)}$ can be explicitly expressed as $\mathbf{\Psi}_u^{(s)}=\sum_{\{u_1\in[1,r],\cdots,u_s\in[1,r]\}} d_{u_1u_2\cdots u_s}\mathbf{\Delta}_{u_1u_2\cdots u_s}=\sum d_{u_1u_2\cdots u_s}\prod_{i=1}^{s}\mathbf{\Gamma}_{i,u_i}=\sum d_{u_1u_2\cdots u_s}\prod_{i=1}^{s}\big\{\sum_{j\in \mathcal{X}_i}\mathbf{C}_{u_i,j}\big\}=\sum d_{u_1u_2\cdots u_s}\prod_{i=1}^{s}\big\{\sum_{j\in [1+(i-1)\lambda,\cdots,i\lambda]}\mathbf{C}_{u_i,j}\big\}=\sum_{\{u_1\in[1,r],\cdots,u_s\in[1,r]\}}d_{u_1u_2\cdots u_s}\sum_{\{j_1\in \mathcal{X}_1,j_2\in \mathcal{X}_2,\cdots,j_s\in \mathcal{X}_s\}}\mathbf{C}_{u_1,j_1}\mathbf{C}_{u_2,j_2}\cdots\mathbf{C}_{u_s,j_s}$, where some of the coefficients $d_{u_1u_2\cdots u_s}$ maybe equal to zero. Henceforth, each non-zero uniterm of $\mathbf{\Psi}_u^{(s)}$ is formed as $d_{u_1u_2\cdots u_s}\mathbf{C}_{u_1,j_1}\mathbf{C}_{u_2,j_2}\cdots\mathbf{C}_{u_s,j_s}$ for $j_{\tau}\in\mathcal{X}_{\tau}=[1+(\tau-1)\lambda,\tau\lambda]$ and $\tau\in[1,s]$. As a result, it is evident from Lemma \ref{premise1} that $\mathbf{S}_i\mathbf{C}_{u_1,j_1}\mathbf{C}_{u_2,j_2}\cdots\mathbf{C}_{u_s,j_s}\preceq\mathbf{S}_i$ for $i\in\mathcal{X}_{s+1}=[s\lambda+1,(s+1)\lambda]$, and therewith we further have that $\mathbf{S}_i\mathbf{\Psi}_u^{(s)}\preceq\mathbf{S}_i$.}.

Operating the above equation on $\mathbf{S}_i$ for each $i$ in $\mathcal{X}_{s+1}=[s\lambda+1,(s+1)\lambda]$, we have
\begin{equation}\label{equiv}
\mathbf{S}_i\mathbf{\Psi}_1^{(s)}+\mathbf{S}_i\mathbf{\Psi}_2^{(s)}(\sum_{j\in\mathcal{X}_{s+1}}\mathbf{C}_{2,j})+\cdots
+\mathbf{S}_i\mathbf{\Psi}_r^{(s)}(\sum_{j\in\mathcal{X}_{s+1}}\mathbf{C}_{r,j})=0,
\end{equation}
which can be rearranged to
$$
\mathbf{S}_i\Big\{\mathbf{\Psi}_1^{(s)}+\mathbf{\Psi}_2^{(s)}(\sum_{j\neq i}\mathbf{C}_{2,j})+\cdots
+\mathbf{\Psi}_r^{(s)}(\sum_{j\neq i}\mathbf{C}_{r,j})\Big\}+\mathbf{S}_i\mathbf{\Psi}_2^{(s)}\mathbf{C}_{2,i}+\cdots
+\mathbf{S}_i\mathbf{\Psi}_r^{(s)}\mathbf{C}_{r,i}=0.
$$
Similarly, Lemma \ref{premise1} leads to that
\begin{equation}
\left\{\begin{aligned}
&\mathbf{S}_i\Big\{\mathbf{\Psi}_1^{(s)}+\mathbf{\Psi}_2^{(s)}(\sum_{j\neq i}\mathbf{C}_{2,j})+\cdots+\mathbf{\Psi}_r^{(s)}(\sum_{j\neq i}\mathbf{C}_{r,j})\Big\}\preceq\mathbf{S}_i;\\
&\mathbf{S}_i\mathbf{\Psi}_{u}^{(s)}\mathbf{C}_{u,i}\preceq\mathbf{S}_i \mathbf{C}_{u,i}, ~~for~~u\in[2,r].
\end{aligned}\right.
\end{equation}

\vspace{0.2cm}
By Lemma \ref{premise2}, we know for each $u\in[2,r]$ and $i\in\mathcal{X}_{s+1}$,
\begin{equation}
\mathbf{S}_i\mathbf{\Psi}_{u}^{(s)}\mathbf{C}_{u,i}=0,
\end{equation}
from which we derive $\mathbf{S}_i\mathbf{\Psi}_{u}^{(s)}=0$ for $\mathbf{C}_{u,i}$ is invertible. Combining with equation (\ref{equiv}), we further know $\mathbf{S}_i\mathbf{\Psi}_1^{(s)}=0$. Thus, for each $u\in[1,r]$ and $i\in\mathcal{X}_{s+1}$, we have $\mathbf{S}_i\mathbf{\Psi}_{u}^{(s)}=0$.

Due to the condition that $\biguplus_{i\in\mathcal{X}_{s+1}}\mathbf{S}_{i}\simeq \mathbb{F}^l$, it has to be that $\mathbf{\Psi}_u^{(s)}=0$ for each $u\in[1,r]$, which contradicts the induction assumption.

\vspace{0.2cm}
\textbf{Conclusion}: That is to say, the $r^{\frac{k}{\lambda}}$ square matrices formed by $\mathbf{\Delta}_{u_1u_2\cdots u_{\frac{k}{\lambda}}}$ are linearly independent. Because they all lie in the $l^2$-dimensional space of matrices $\mathbb{F}^{l\times l}$, we have
\begin{equation}
r^{\frac{k}{\lambda}}\leq l^2\Leftrightarrow k\leq 2\lambda\log_rl.
\end{equation}

\textbf{2. The situation when $\lambda$ does not divide $k$.}
\vspace{0.2cm}

Assume $p\lambda<k<(p+1)\lambda$, i,e., $p=\left \lfloor\frac{k}{\lambda}\right \rfloor$ and $p+1=\left\lceil\frac{k}{\lambda}\right \rceil$. Then, divide $[1,k]$ into $p+1$ disjoint sets as $\mathcal{X}_1=[1,\cdots,k-p\lambda]$ and $\mathcal{X}_i=[k-(p+2-i)\lambda+1,\cdots,k-(p+1-i)\lambda]$, where $|\mathcal{X}_1|=k-p\lambda$ and $|\mathcal{X}_i|=\lambda$ for each $i\in[2,p+1]$. It should be noted here that $\mathcal{X}_1$ is not a standard partition and repair subspaces within it cannot span the whole space $\mathbb{F}^l$, because $|\mathcal{X}_1|<\lambda$.
\vspace{0.1cm}

As in the first situation, we only need to consider the base case, since we apply the same method of induction. In this case, by operating on $\mathbf{S}_i$ for any $i\in[1,k-p\lambda]$, it is trivial to verify that
\begin{equation}
\Big\{\mathbf{\Gamma}_{1,u_1}\mid u_1\in[1,r]\Big\}
=\left\{\mathbf{I},\sum_{j=1}^{k-p\lambda}\mathbf{C}_{2,j},\cdots,
\sum_{j=1}^{k-p\lambda}\mathbf{C}_{r,j}\right\}
\end{equation}
are also linearly independent. Thus, we totally obtain $r^{p+1}$ linearly independent square matrices
$$\left\{\mathbf{\Delta}_{u_1u_2\cdots u_{p+1}}=\prod_{i=1}^{p+1}\mathbf{\Gamma}_{i,u_i}\mid u_i\in[1,r],i\in[1,p+1]\right\}.$$

Therefore, we derive
\begin{equation}
r^{\left\lceil\frac{k}{\lambda}\right \rceil}\leq l^2\Rightarrow k< 2\lambda\log_rl.
\end{equation}$\hfill\blacksquare$

\end{proof}

\begin{rem}
In \cite{Re:S.Goparaju1}, Goparaju et al. derive that $k\leq 2\lambda\log_2l$, where $\lambda\leq\left\lfloor\log_{\frac{r}{r-1}}l\right \rfloor+1$. Here, we improve the upper bound from log-base $2$ to log-base $r$ and obtain that $k\leq2\lambda\log_rl$.

\end{rem}

In subsequent, we estimate the size of each standard partition $\lambda$, with which we will further give an explicit upper bound.

\subsection{An Explicit $\frac{r^2}{l}$-Dependent Bound}

In the following, we use our previous work \cite{Kun} to estimate the size of standard partition $\lambda$, combining which with Theorem \ref{log} we will further present an explicit bound depending on the value of $\frac{r^2}{l}$.

\subsubsection{3.3.1 Estimation of $\lambda$.}

In \cite{Re:S.Goparaju1}, Goparaju et al. apply the geometric analysis of the invariant subspace and derive that for any $m\in [1,k]$,
\begin{equation}\label{pre}
\dim(\biguplus_{i=1}^{m}\mathbf{S}_i)=\dim(\mathbf{S}_1\uplus\mathbf{S}_2\uplus\cdots\uplus\mathbf{S}_m)\geq \big(1-(\frac{r-1}{r})^{m}\big)l,
\end{equation}
from which they obtain that $\lambda\leq\left\lfloor\log_{\frac{r}{r-1}}l\right \rfloor+1$, i.e., if the size of a standard partition equals with $\left\lfloor\log_{\frac{r}{r-1}}l\right \rfloor+1$, the repair subspaces within it necessarily span the entire space $\mathbb{F}^l$. In addition, they find that formula (\ref{pre}) is useful for studying the secrecy capacity of linear MSR codes in \cite{Re:S.Goparaju}. We also focus on this secrecy issue in \cite{Kun} and ultimately present an explicit result on secrecy capacity in the linear MSR scenario that extends the result given in \cite{Re:S.Goparaju}. As shown in Theorem $6$ of \cite{Kun}, it closely depends on the value of $\beta$ that is equal to $\frac{l}{r}$ herein. The following theorem can be looked upon as a direct corollary of formula $(72)$ of Theorem $6$ given in \cite{Kun}.

\begin{theorem}\label{class}
Given an $\{n=k+r,k,l\}$ linear systematic-repair MSR code with independent repair subspaces, for any set $F\subsetneq[1,k]$ with $|F|=m$, we have
\begin{equation}\label{part}
\dim(\biguplus_{i\in F}\mathbf{S}_i):\left\{\begin{array}{ll}
=m\frac{l}{r}, &\textrm{if} \quad m\leq t;\\
\geq t\frac{l}{r}+\frac{l}{r}(r-t)[1-(\frac{r-1}{r})^e]= l-\frac{l}{r}(r-t)(\frac{r-1}{r})^{e}, &\textrm{if} \quad m=t+e,
\end{array}\right.
\end{equation}
where $t=\lceil\frac{r^2}{l}\rceil$ and $e\geq 1$.
\end{theorem}

\begin{proof}

As stated in Remark $11$ of \cite{Kun}, all results therein in fact also apply to MSR codes with optimal repair of systematic nodes only, although the general case for MSR codes with optimal repair of all nodes is the study object in \cite{Kun}. In this paper, we let a linear systematic-repair MSR code represented by $\{n=k+r,k,l\}$, where $n=d+1$, $l=\alpha$ and $\frac{l}{r}=\beta$. Through changing the expression of related parameters, formula $(72)$ of Theorem $6$ in \cite{Kun} can be equivalently transformed into as follows
\begin{equation}\label{part}
\dim(\biguplus_{i\in F}\mathbf{S}_i):\left\{\begin{array}{ll}
=m\frac{l}{r}, & if \quad m\leq t, \frac{l}{r}<\frac{r}{t-1};\\
\geq t\frac{l}{r}+\frac{l}{r}(r-t)[1-(\frac{r-1}{r})^e]= l-\frac{l}{r}(r-t)(\frac{r-1}{r})^{e}, & if \quad m=t+e, \frac{r}{t}\leq \frac{l}{r}<\frac{r}{t-1},
\end{array}\right.
\end{equation}
where $1\leq t\leq r$ and $e\geq 1$. Since $t$ is an integer satisfying the condition that $\Big\{\frac{r}{t}\leq \frac{l}{r}<\frac{r}{t-1}\Leftrightarrow\frac{r^2}{l}\leq t<\frac{r^2}{l}+1\Big\}$, we can rewrite it as $t=\lceil\frac{r^2}{l}\rceil$, where it is trivial that $1\leq t\leq r$.$\hfill\blacksquare$

\end{proof}

\begin{rem}\label{verify}
It has been verified in \cite{Kun} that,
\begin{equation}
\left\{\begin{aligned}
&t\frac{l}{r}+\frac{l}{r}(r-t)[1-(\frac{r-1}{r})^e]<t\frac{l}{r}+\frac{l}{r}(r-t)[e(1-\frac{r-1}{r})]
=t\frac{l}{r}+e\frac{l}{r}(1-\frac{t}{r})<(t+e)\frac{l}{r}=m\frac{l}{r};\\
&t\frac{l}{r}+\frac{l}{r}(r-t)[1-(\frac{r-1}{r})^e]=l[1-(\frac{r-t}{r})(\frac{r-1}{r})^e]
\geq l[1-(\frac{r-1}{r})^t(\frac{r-1}{r})^e]=[1-(\frac{r-1}{r})^{m}]l,\\
\end{aligned}\right.
\end{equation}
which implies that when $\dim(\biguplus_{i\in F}\mathbf{S}_i)\leq l$,
\begin{equation}\label{compare}
[1-(\frac{r-1}{r})^{m}]l\leq\dim(\biguplus_{i\in F}\mathbf{S}_i)\leq m\frac{l}{r}.
\end{equation}

As assumed in \cite{Re:S.Goparaju1}, we also make hypothesis that $\lambda< k$, i.e., there exists at least one standard partition properly included in $[1,k]$. Otherwise, we have $k\leq\lambda$. Next, we apply Theorem \ref{class} to estimate $\lambda$, the smallest value of $m$ satisfying $\dim(\biguplus_{i\in F}\mathbf{S}_i)=l$.

\end{rem}

\begin{theorem}\label{partition}
Given an $\{n=k+r,k,l\}$ linear systematic-repair MSR code with independent repair subspaces,
we have
\begin{equation}
\lambda:\left\{\begin{array}{ll}
=r, &\textrm{if} \quad t=r;\\
\leq t+1+\left \lfloor\log_{\frac{r}{r-1}}\frac{(r-t)l}{r}\right \rfloor, &\textrm{if} \quad 1\leq t<r.
\end{array}\right.
\end{equation}
where $t=\lceil\frac{r^2}{l}\rceil$.
\end{theorem}
\begin{proof}

There are two cases analyzed as follows.
\vspace{0.2cm}

\textbf{1. The case when $t=r$.}

According to the first item of Theorem \ref{class}, we know when $m=r$, it is clear that $\dim(\biguplus_{i\in F}\mathbf{S}_i)=l$. Thus, we have $\lambda=r$.
\vspace{0.2cm}

\textbf{2. The case when $1\leq t<r$.}

According to the second item of Theorem \ref{class}, we know when $m=t+e$, then $\dim(\biguplus_{i\in F}\mathbf{S}_i)\geq l-\frac{l}{r}(r-t)(\frac{r-1}{r})^{e}$. Hence, when $\frac{l}{r}(r-t)(\frac{r-1}{r})^{e}<1$ or $e>\log_{\frac{r}{r-1}}\frac{(r-t)l}{r}$, it has to be that $\dim(\biguplus_{i\in F}\mathbf{S}_i)=l$. Thus, we obtain that $\lambda\leq t+e=t+1+\left \lfloor\log_{\frac{r}{r-1}}\frac{(r-t)l}{r}\right \rfloor$.$\hfill\blacksquare$
\end{proof}

\begin{rem}

From Theorem \ref{partition}, we know when $t=1$, it is clear that $\lambda\leq2+\left \lfloor\log_{\frac{r}{r-1}}\frac{(r-1)l}{r}\right \rfloor=2+\left \lfloor\log_{\frac{r}{r-1}}l-1\right \rfloor=1+\left \lfloor\log_{\frac{r}{r-1}}l\right \rfloor$, which is consistent with the corresponding result in \cite{Re:S.Goparaju1}. Besides, one can check that, for $1<t<r$,
\begin{equation}
r<t+1+\left \lfloor\log_{\frac{r}{r-1}}\frac{(r-t)l}{r}\right \rfloor<1+\left \lfloor\log_{\frac{r}{r-1}}l\right \rfloor,
\end{equation}
which basically corresponds to the inequality (\ref{compare}).
\end{rem}
\vspace{0.1cm}

\subsubsection{3.3.2 Final Explicit Result.}

Eventually, combining Theorem \ref{log} and Theorem \ref{partition}, we derive the final explicit bound as follows.

\begin{theorem}\label{final}
For any $\{n=k+r,k,l\}$ linear systematic-repair MSR code with independent repair subspaces, the following upper bound holds
\begin{equation}
k\leq 2\lambda\log_rl,
\end{equation}
wherein
\begin{equation}
\lambda:\left\{\begin{array}{ll}
=r, &\textrm{when} \quad t=r;\\
\leq t+1+\left \lfloor\log_{\frac{r}{r-1}}\frac{(r-t)l}{r}\right \rfloor, &\textrm{when} \quad 1\leq t<r,
\end{array}\right.
\end{equation}
where $t=\lceil\frac{r^2}{l}\rceil$.
\end{theorem}

\begin{rem}

Through literature survey, MSR codes for $\{d=n-1\}$ known so far are generally divided into scalar MSR codes with $\{l=r\}$ and vector MSR codes with $\{l=r^{x+1},x\geq1\}$, which correspond to $t=r$ and $t=1$ respectively. As stated in \cite{Kun}, it remains open for constructing vector MSR codes with $\{r<l<r^2\}$.

When $l=r$, Theorem \ref{scalar} implies that $k=r+1$ is the longest systematic-length for the scalar linear systematic-repair MSR codes with independent repair subspaces. According to Theorem \ref{final}, we know that when $l=r$, it has to be that $\lambda=t=r$, following which we have $k\leq 2r$. Essentially, when $k=r+1$, it is evident that $\lambda=r$ does not divide $k$ and thus it should be that $k<2r$ following from Theorem \ref{log}. The reason is that, for $k=r+1$, $k$ systematic nodes can only be split into one standard partition as well as an nonstandard partition, where the size of the nonstandard partition equals with $1$ that is smaller than $r$. In some sense, the $r$-based-log upper bound basically coincides with the real largest number of systematic nodes in the scalar scenario. In other words, the longest systematic-length $k=r+1$ is indeed included in the actual case of $k<2r$ derived from the $r$-based-log upper bound for $l=r$.

When $l\geq r^2$, it is clear that $t=1$. For example in Table \ref{Tab:Comparision}, the systematic-repair MSR code with the longest systematic-length known so far given in \cite{Re:Z.Wang2} is with $\big\{k=(r+1)\log_rl\big\}$. However, Theorem \ref{final} leads to that $k\leq 2\log_rl\big(1+\left \lfloor\log_{\frac{r}{r-1}}l\right \rfloor\big)$. By comparison, we find that $r+1$ is strictly less than $2\big(1+\left \lfloor\log_{\frac{r}{r-1}}l\right \rfloor\big)$ even for $l=r^2$ and $r=2$, which means that the code in \cite{Re:Z.Wang2} cannot reach the derived upper bound. That is to say, although we do improve the upper bound from previous log-base $2$ to current log-base $r$, it is still unknown yet whether the upper bound in Theorem \ref{final} can be achieved for some unexplored vector linear systematic-repair MSR code. Nevertheless, our $r$-based-log upper bound is nearly in line with the conjecture that $k$ is of the order of $\log_rl$ proposed by Tamo et al. in \cite{Re:I.Tamo1}.
\end{rem}

\subsection{Further Discussions}

As stated in Remark \ref{verify}, our Theorem \ref{class}, Theorem \ref{partition} and Theorem \ref{final} all are based on the assumption $\lambda<k$, i.e., systematic part at least properly includes one standard partition, while Theorem \ref{log} also applies to the case of $k=\lambda$ as the systematic part $[1,k]$ exactly constitutes a standard partition. Thereby, the situation $k\leq\lambda$ is left open.

\vspace{0.1cm}

When $k=\lambda$, it is clear that $k\leq 2\lambda\log_rl$ from Theorem \ref{log}, since it is trivial that $r\leq l$ for $r$ divides $l$. When $k<\lambda$, the systematic part cannot satisfy the condition for the formation of a standard partition, which implies the notation $\lambda$ is meaningless at present. In this case, we cannot obtain the log-based upper bound similarly. Hence, we put forward two questions as follows.

\begin{question}\label{1}
\emph{What value of $l$ and $r$ taken can we make the largest number of systematic node $k$ equal with $\lambda$ in the linear systematic-repair MSR scenario? What is the exact value of $k$ for this case?}
\end{question}

Technically, under the condition (\ref{interference}), we assume $k_{\{l,r\}}$ is the longest systematic-length for given $l$ and $r$. Then, what values of $l$ and $r$ taken will make $k_{\{l,r\}}$ satisfy

\begin{equation}
\left\{\begin{aligned}
&\dim(\biguplus_{i=1}^{k_{\{l,r\}}} \mathbf{S}_{i})=l;\\
&\dim(\biguplus_{j=1}^{k_{\{l,r\}}-1} \mathbf{S}_{i_j})<l, for~ any~ distinct~~ i_j\in[1,k_{\{l,r\}}],
\end{aligned}\right.
\end{equation}
where $\dim(\mathbf{S}_{i})=\frac{l}{r}$ for each $i\in[1,k_{\{l,r\}}]$.

\begin{question}\label{2}
\emph{What value of $l$ and $r$ taken will make the systematic part cannot form a standard partition in the linear systematic-repair MSR scenario? What is the largest number of systematic nodes $k$ now?}
\end{question}

With the same definition of $k_{\{l,r\}}$ as above, what range of values for $l$ and $r$ will make $k_{\{l,r\}}$ satisfy

\begin{equation}
\dim(\biguplus_{i=1}^{k_{\{l,r\}}} \mathbf{S}_{i})<l,
\end{equation}
where $\dim(\mathbf{S}_{i})=\frac{l}{r}$ for each $i\in[1,k_{\{l,r\}}]$.

\begin{rem}
Question \ref{1} and Question \ref{2} are our next direction of study. It may help further improve the upper bound on systematic-length, since our $r$-based-log upper bound are built on the assumption $\dim(\biguplus_{i=1}^{k_{\{l,r\}}}\mathbf{S}_{i})>l$. The intuition establishes on the fact that the length $k_{\{l,r\}}$ of the systematic part herein is not larger than $\lambda$ the least number of systematic nodes for supporting a standard partition.

\end{rem}

\section{CONCLUSION}

Finding the exact systematic-length upper bound for a systematic-repair MSR code with $r$ parity nodes and storage capacity $l$ is an open problem. Following the method of geometric analysis on a set of subspaces and operators, we formalize two helpful derivative properties of subspaces under the case of arbitrary number of parity nodes. With them, we first demonstrate the linear independence among all encoding matrices and then design new linearly independent matrices based on the idea of partition, which naturally bound the systematic-length $k$. Finally, we derive a simple quadratic bound and an $r$-based-log bound, which both are superior to the previous results. Moreover, leveraging our prior work \cite{Kun}, we estimate the size of a standard partition and further present an explicit bound depending on the value of $\frac{r^2}{l}$.

\section{APPENDIX}

\begin{corollary}
In Theorem \ref{log}, we define
$$\mathbf{\Delta}_{u_1u_2\cdots u_{\frac{k}{\lambda}}}=\prod_{i=1}^{\frac{k}{\lambda}}\mathbf{\Gamma}_{i,u_i},$$
where
\begin{equation}
\mathbf{\Gamma}_{i,u}=\left\{\begin{aligned}
&\mathbf{I},\quad if ~ u=1;\\
&\sum_{j\in \mathcal{X}_i}\mathbf{C}_{u,j}, \quad if~ u\in[2,r],
\end{aligned}\right.
\end{equation}
for $u_i\in[1,r]$ and $i\in[1,\frac{k}{\lambda}]$. In fact, all of these $r^{\frac{k}{\lambda}}$ matrices are non-zero.
\end{corollary}

\begin{proof}
We still employ the proof of contradiction.
\vspace{0.2cm}

\textbf{Induction claim}: Assume for some $s\in [1,\frac{k}{\lambda}]$, the following square matrix
$$\mathbf{\Delta}_{u_1u_2\cdots u_s}=\prod_{i=1}^{s}\mathbf{\Gamma}_{i,u_i}\neq 0,$$
where $u_i\in[1,r]$.

\vspace{0.2cm}
\textbf{Base case}: For $s=1$, we have
$$\mathbf{\Delta}_{u_1}=\mathbf{\Gamma}_{1,u_1},$$
where it is clear from the definition that $\mathbf{\Gamma}_{1,u_1}=\mathbf{I}\neq 0$ if $u_1=1$.

\vspace{0.2cm}

If $u_1\in[2,r]$, $\mathbf{\Gamma}_{1,u_1}=\sum_{j\in \mathcal{X}_1}\mathbf{C}_{u_1,j}$. According to Theorem \ref{scalar}, we know that the matrices $\{\mathbf{C}_{u_1,j}\mid j\in \mathcal{X}_1\}$ are linearly independent, from which it is apparent that $\mathbf{\Gamma}_{1,u_1}=\sum_{j\in \mathcal{X}_1}\mathbf{C}_{u_1,j}\neq 0$ for any $u_1\in[2,r]$.

\vspace{0.2cm}
\textbf{Inductive step}: Let the inductive claim hold for some $s$, then it is true for $s+1$. Otherwise, we have
\begin{equation}
\mathbf{\Delta}_{u_1u_2\cdots u_{s+1}}=\prod_{i=1}^{s+1}\mathbf{\Gamma}_{i,u_i}=0.
\end{equation}

1. When $u_{s+1}=1$, we have $\mathbf{\Gamma}_{s+1,1}=\mathbf{I}$. Thus, it has to be that $\mathbf{\Delta}_{u_1u_2\cdots u_{s+1}}=\mathbf{\Delta}_{u_1u_2\cdots u_{s}}=0$, which contradicts the induction assumption.
\vspace{0.2cm}

2. When $u_{s+1}\in[2,r]$, $\mathbf{\Gamma}_{s+1,u_{s+1}}=\sum_{j\in \mathcal{X}_{s+1}}\mathbf{C}_{u_{s+1},j}$, with which we have
\begin{equation}
\mathbf{\Delta}_{u_1u_2\cdots u_{s+1}}=\mathbf{\Delta}_{u_1u_2\cdots u_{s}}(\sum_{j\in \mathcal{X}_{s+1}}\mathbf{C}_{u_{s+1},j})=0.
\end{equation}
Operating the above equation on $\mathbf{S}_i$ for each $i$ in $\mathcal{X}_{s+1}=[s\lambda+1,(s+1)\lambda]$, we obtain
\begin{equation}
\mathbf{S}_i\mathbf{\Delta}_{u_1u_2\cdots u_{s}}\mathbf{C}_{u_{s+1},i}+\mathbf{S}_i\mathbf{\Delta}_{u_1u_2\cdots u_{s}}(\sum_{j\neq i}\mathbf{C}_{u_{s+1},j})=0,
\end{equation}
where $\mathbf{S}_i\mathbf{\Delta}_{u_1u_2\cdots u_{s}}\mathbf{C}_{u_{s+1},i}\preceq\mathbf{S}_i \mathbf{C}_{u_{s+1},i}$ and $\mathbf{S}_i\mathbf{\Delta}_{u_1u_2\cdots u_{s}}(\sum_{j\neq i}\mathbf{C}_{u_{s+1},j})\preceq\mathbf{S}_i$ from Lemma \ref{premise1}.

\vspace{0.2cm}
Further by Lemma \ref{premise2}, we derive, for each $i\in\mathcal{X}_{s+1}$,
\begin{equation}
\mathbf{S}_i\mathbf{\Delta}_{u_1u_2\cdots u_{s}}\mathbf{C}_{u_{s+1},i}=0,
\end{equation}
from which we have $\mathbf{S}_i\mathbf{\Delta}_{u_1u_2\cdots u_{s}}=0$, since $\mathbf{C}_{u_{s+1},i}$ is invertible. Under the condition $\biguplus_{i\in\mathcal{X}_{s+1}}\mathbf{S}_{i}\simeq \mathbb{F}^l$, it must be that $\mathbf{\Delta}_{u_1u_2\cdots u_{s}}=0$, which also contradicts the induction assumption.
\end{proof}

\vspace{0.2cm}
\textbf{Conclusion}: That is to say, any square matrix formed as $\mathbf{\Delta}_{u_1u_2\cdots u_{\frac{k}{\lambda}}}$ is non-zero.$\hfill\blacksquare$

\end{document}